\documentclass[10pt,a4paper]{article}
\usepackage{latexsym,amsthm,amssymb}

\usepackage[dvips]{graphics}
\usepackage{amssymb, amsmath, amsthm}
\usepackage{graphicx}
\newtheorem{theorem}{Theorem}
\newtheorem{proposition}[theorem]{Proposition}
\newtheorem{corollary}[theorem]{Corollary}

\newcommand{\be}{\begin{equation}}
\newcommand{\ee}{\end{equation}}
\newcommand{\bea}{\begin{eqnarray}}
\newcommand{\eea}{\end{eqnarray}}
\newcommand{\ba}{\begin{array}}
\newcommand{\ea}{\end{array}}
\newcommand{\bean}{\begin{eqnarray*}}
\newcommand{\eean}{\end{eqnarray*}}

\newcommand{\om}{\omega}

\newcommand{\sig}{\sigma}

\newcommand{\pa}{\partial}

\newcommand{\res}{{\rm res\/}}
\newcommand{\wh}{\widehat}

\begin{document}

\title
 {\sc A note on the extended dToda hierarchy\/}
\author
{\sc Niann-Chern Lee$^1$ and Ming-Hsien Tu$^2$
\footnote{phymhtu@ccu.edu.tw} \\
  $^1$
  {\it General Education Center, National Chin-Yi University of Technology\/},\\
  {\it Taichung 411, Taiwan\/}\\
   $^2$
  {\it Department of Physics, National Chung Cheng University\/},\\
   {\it Chiayi 621, Taiwan\/}\/}
\date{\today}
\maketitle
\begin{abstract}
We give a derivation of dispersionless Hirota equations for the
extended dispersionless Toda hierarchy. We show that the dispersionless Hirota
equations are nothing but a direct consequence of the genus-zero topological recursion
relation for the topological $CP^1$ model. Using the dispersionless Hirota equations we compute
the two point functions and express the result in terms of Catalan number.
\end{abstract}

\newpage

\section{Introduction}
Recently, Kodama and Pierce\cite{KP09} gave a combinatorial description of  the
 one-dimensional dispersionless Toda(dToda) hierarchy to solve the two-vertex problem on a sphere.
 The main strategy is to characterize  the free energy $F(t_0,t)$ ($t=(t_1,t_2,\cdots)$) of the dToda hierarchy by the
 corresponding dispersionless Hirota equations.
Then the second derivatives of the free energy $\pa_{t_n}\pa_{t_m}F\equiv F_{n,m}$ satisfy a set of algebraic relations.
Surprisingly they found a closed form for the rational numbers $F_{n,m}$ under the conditions $F_{01}=F_{00}=0$
for general $n$ and $m$. In particular, the formulas of $F_{n,m}$ can be expressed in terms of the Catalan number
which is commonly used in the context of enumerative combinatorics (see e.g. \cite{Stanley00}).
Their result for $F_{n,m}$  provides a combinatorial meaning of a counting problem of connected ribbon graphs
with two vertices of degree $n$ and $m$ on a sphere
and is a generalization of the previous works where the problem has been solved only
in the case of the same degree (that is $F_{nn}$) \cite{J93,KL93}.

In this work, motivated by the aforementioned result, we like to generalize the computation of the two point functions
$F_{n,m}$ to the extended dToda hierarchy\cite{EY94,EHY95,DZ01,DZ04}
which is an extension of the  one-dimensional dToda hierarchy by adding logarithmic type conserved densities.
Since extended dToda hierarchy is the dispersionless limit
of the extended Toda hierarchy \cite{Z02,CDZ04} which has been used to govern the
Gromov-Witten(GW) invariants(see e.g. \cite{HKKPTVVZ03} and references therein) for the $CP^1$ manifold.
Thus the extended dToda hierarchy becomes the master equation of the genus zero GW invariants whose generating
function is characterized by the free energy of the extended dToda hierarchy.
Based on the twistor theoretical method \cite{TT95,KO95} the extended dToda hierarchy can be constructed
by adding logarithmic-flow to the one-dimensional dToda hierarchy.
The corresponding Orlov-Schulman operator is conjugated with the Lax operator
under the Poisson bracket which imposes an extra condition (the so-called string equation)
 on the free energy of the extended dToda hierarchy.
 We will show that the full hierarchy flows  can be expressed in terms of second derivatives of its
 associated free energy $F$ and thus  can be viewed as the corresponding dispersionless Hirota(dHirota) equations.
We then investigate the two point functions of the extended dToda hierarchy
based on the associated dHirota equations and express the result in terms of the Catalan number.
To make a connection with the topological field theory, we rewrite the dHirota equations in $CP^1$ time parameters
and show that they are indeed a direct consequence of the genu-zero topological recursion relation\cite{W90}
of the topological $CP^1$ model.

This paper is organized as follows. In section 2, we recall the Lax formalism of
 the extended dToda hierarchy. In section 3, we derive the dHirota equation of the extended dToda hierarchy
which can be expressed as a set of equations in terms of second derivatives of the free energy.
The initial values of two-point functions of the extended dToda hierarchy are computed in Section 4.
In section 5, we reinterpret  the dHirota equations from topological field theory point of view.
 Section 6 is devoted to the concluding remarks.

\section{The extended dispersionless Toda hierarchy}
The one-dimensional dToda hierarchy\cite{TT95,KP09} is defined by the Lax equation
\[
\frac{\pa L}{\pa t_n}=\{B_n, L\},\quad B_n=(L^n)_{\ge 0}.
\label{Lax-eq}
\]
where $L$ ia a two-variable Lax operator of the form
 \be
  L=p+u_1+u_2p^{-1}
  \label{Laxop-2}
 \ee
with $u_1$ and $u_2$ are functions of the time variables $t=(t_1,t_2,\ldots)$
along with a spatial variable $t_0$.
Here $(A)_{\geq 0}$ denotes the polynomial part of $A$,
$(A)_{\leq -1}=A-(A)_{\geq 0}$, and the Poisson bracket $\{,\}$ is defined by
\[
 \{A(p,t_0),B(p,t_0)\} = p\frac{\pa A(p,t_0)}{\pa p}\frac{\pa B(p,t_0)}{\pa t_0}
                       - p\frac{\pa A(p,t_0)}{\pa t_0}\frac{\pa B(p,t_0)}{\pa p}.
\]
In particular, the fundamental variable $u_1$ and $u_2$
can be expressed in terms of second derivatives of $F$ as
\[
u_1=F_{01},\quad u_2=F_{11}=e^{F_{00}}
\]
where the second equation is just the one-dimensional reduction of the dToda field equation.
Following the twistor theoretical construction \cite{TT95,KO95}, the extended dToda hierarchy can be
constructed from the one-dimensional dToda hierarchy by adding the  $\hat{t}_n$-flows as
\be
 \frac{\pa L}{\pa \hat{t}_n} = \{\hat{B}_n, L\},\quad \hat{B}_n=\left(L^n(\log L-d_n)\right)_{\geq 0}
 \label{eLaxeq}
\ee
where $d_n=\sum_{j=1}^n1/j$ with $d_0\equiv 0$ and $\log L$ is defined by the prescription
\be
\log L=\frac{1}{2}\log u_2+\frac{1}{2}\log(1+u_1p^{-1}+u_2p^{-2})+
\frac{1}{2}\log\left(1+\frac{u_1}{u_2}p+\frac{1}{u_2}p^2\right)
\label{logL-exp}
\ee
with the proviso that we shall Taylor expand the second term in
$p^{-1}$, whereas in $p$ for the last term. Moreover, the associated
Orlov-Schulman is given by
\[
 N(t_0,t,\hat{t})
= \sum_{n=1}nt_n L^n +t_0 +\sum_{n=1}n\hat{t}_n L^n(\log
L-d_{n-1})+\sum_{n=1}F_{n0}L^{-n},
 \]
 which satisfies
\[
 \pa_{t_n} N=\{B_n,N\}, \quad \pa_{\hat{t}_n}N=\{\hat{B}_n,N\},\quad \{L, N\}=L.
\]
The symplectic two-form of the extended dToda hierarchy can be written as
\[
 \om \equiv \frac{dp}{p}\wedge dt_0
     +\sum_{n=1}^{\infty}dB_n\wedge dt_n
     +\sum_{n=1}^{\infty}d\hat{B}_n\wedge d\hat{t}_n=\frac{d L\wedge dN}{ L}
\]
which implies the existence of a $S$ function such that
\[
dS(t_0,t,\hat{t})=Nd\log L+\log p dt_0 +\sum_{n=1}^{\infty}B_n dt_n
+\sum_{n=1}^{\infty}\hat{B}_n d\hat{t}_n
\]
or, equivalently,
 \[
 N=\frac{\pa S}{\pa \log L},\quad \log
p=\frac{\pa S}{\pa t_0},\quad B_n=\frac{\pa S}{\pa t_n},\quad
\hat{B}_n=\frac{\pa S}{\pa \hat{t}_n}.
\]
It is not hard to show that the $S$ function has the form
 \[
S=\sum_{n=1}t_n L^n+t_0\log L+\sum_{n=1}\hat{t}_nL^n(\log L-d_n)-
\sum_{n=1}\frac{F_{0n}}{n}L^{-n}.
 \]
 Setting $\hat{t}_n=0$ for $n\geq 1$, it recovers the $S$ function of the one-dimensional dToda hierarchy.
Finally, the twistor construction\cite{KO95} enables us to extract the string equation
\be
 -1= \sum_{n=2}^{\infty}nt_n\frac{\pa L}{\pa t_{n-1}} +
\sum_{n=1}^{\infty}n\hat{t}_n\frac{\pa L}{\pa \hat{t}_{n-1}}.
\label{string-eq}
 \ee
 for the extended dToda hierarchy without referring to the $CP^1$ matrix model\cite{EY94}.
\section{Dispersionless Hirota equations}
\begin{proposition} The following relations hold.
\bea
 F_{n0}&=&(B_n)_{[0]},\quad F_{n1}=\res(L^n),\quad n\geq 1
 \label{bF1}\\
F_{\wh{n}0}&=&(\wh{B}_n)_{[0]},\quad F_{\wh{n}1}=\res(L^n(\log
L-d_n)),\quad n\geq 0
\label{bF2}
 \eea
 where $(\sum_ka_kp^k)_{[j]}=a_j$.
\end{proposition}
 \begin{proof}
From $\log p=\pa S/\pa t_0$ we have
 \be
  \log p=\log
L-\sum_{n=1}\frac{F_{0n}}{n}L^{-n}
 \label{HJ1}
 \ee
  or
\[
L=p e^{\sum_{n=1}F_{0n}L^{-n}}=p+F_{01}+
\left(-u_1F_{01}+\frac{1}{2}F_{02}+\frac{1}{2}(F_{01})^2\right)p^{-1}+O(p^{-2})
\]
which yields $u_1=F_{01}$ and
$u_2=\frac{1}{2}F_{02}-\frac{1}{2}(F_{01})^2$. Therefore, from the
$p^0$-term of the Lax equation (\ref{Lax-eq}) we have
$F_{n0}=(B_n)_{[0]}$.
On the other hand, from $B_n=\pa S/\pa t_n$ we have
 \[
B_n=L^n_+=L^n-\sum_{m=1}\frac{F_{nm}}{m}L^{-m}.
 \]
  For
$n=1$, we have $u_2p^{-1}=\sum_{m=1}F_{1m}L^{-m}/m$ which together with
(\ref{HJ1}) implies $ F_{m1}=mu_2P_{m-1}(F_{0j}/j)$  where $P_m(t)$ are Schur polynomials defined by
 $e^{\sum_{j=1}t_jz^j}=\sum_{j=0}P_j(t)z^j$.
 In particular, $u_2=e^{F_{00}}=F_{11}$. Also, for the $p^{-1}$-term of the Lax equation (\ref{Lax-eq}) we have
 $F_{n1}=\res(L^n)=u_2(B_n)_{[1]}$.
Furthermore, from $\hat{B}_n=\pa S/\pa \hat{t}_n$ we have
 \[
\hat{B}_n=[L^n(\log L-d_n)]_+=L^n(\log
L-d_n)-\sum_{m=1}\frac{F_{\hat{n}m}}{m}L^{-m}.
   \]
The $p^{-1}$-term gives $F_{\hat{n}1}=\res(L^n(\log L-d_n))=u_2(L^n(\log L-d_n))_{[1]}$
where the last equality is due to the the identity
$\mbox{res}(L^n(\log L-d_n)dL)=0$. Finally, from the $p^0$-term of
the Lax equation (\ref{eLaxeq}) we have $F_{\hat{n}0}=(\hat{B}_n)_{[0]}$.
\end{proof}
\begin{proposition} \label{Fn01}
The two point functions $F_{n0}$, $F_{n1}$, $F_{\hat{n}0}$, and
$F_{\hat{n}1}$ can be expressed in terms of $F_{01}$ and $F_{00}$
as follows
 \bea
 F_{n0}&=&\sum_{s=0}^{[\frac{n}{2}]}\frac{n!}{s!s!(n-2s)!}F_{01}^{n-2s}e^{sF_{00}}
 \label{Fn0}\\
 F_{n+1,1}&=&\sum_{s=0}^{[\frac{n}{2}]}\frac{(n+1)!}{s!(s+1)!(n-2s)!}F_{01}^{n-2s}e^{(s+1)F_{00}}
 \label{Fn1}\\
F_{\wh{n}0}&=&\frac{1}{2}\sum_{s=0}^{[\frac{n}{2}]}\frac{n!}{s!s!(n-2s)!}F_{01}^{n-2s}e^{sF_{00}}(F_{00}-2d_s)
\label{Fhn0}\\
F_{\wh{n+1},1}&=&\frac{1}{2}\sum_{s=0}^{[\frac{n}{2}]}\frac{(n+1)!}{s!(s+1)!(n-2s)!}F_{01}^{n-2s}
e^{(s+1)F_{00}}\left(F_{00}-2d_s-\frac{1}{s+1}\right).
 \label{Fhn1}
 \eea
\end{proposition}
\begin{proof}
Using the binomial expansion of powers of $L$ in (\ref{bF1}) and the
Taylor expansion in (\ref{bF2}) with the prescription (\ref{logL-exp}) for $\log L$.
\end{proof}
We come now to the main result of the work; that is to derive the dHirota equation
for the extended dToda hierarchy from the Lax formulation. The result will be expressed in terms of
second derivatives of the free energy $F(t_0,t,\hat{t})$.
\begin{theorem} \label{F123} The free energy $F(t_0,t,\hat{t})$  of the extended dToda hierarchy satisfies the following equations
\bea
&&\frac{F_{n+1,m}}{n+1}+\frac{F_{n,m+1}}{m+1}=F_{m,0}F_{n,1}+
F_{m,1}F_{n,0},\quad (n\geq 1,m\geq 1)
\label{F1}\\
&&\frac{F_{\wh{n+1},m}}{n+1}+\frac{F_{\hat{n},m+1}}{m+1}=F_{m,0}F_{\hat{n},1}+
F_{m,1}F_{\hat{n},0},\quad (m\geq 1,n\geq 0) \label{F2} \\
&&\frac{F_{\wh{n+1},\wh{m}}}{n+1}+\frac{F_{\hat{n},\wh{m+1}}}{m+1}=F_{\wh{m},0}F_{\hat{n},1}+
F_{\wh{m},1}F_{\hat{n},0},\quad (m,n\geq 0) \label{F3}.
 \eea
\end{theorem}
\begin{proof}
To prove (\ref{F1}), we note that
 \bean
 F_{m,n+1,0}&=&\frac{\pa F_{n+1,0}}{\pa t_{m}}
 =\frac{\pa (L^{n+1})_{[0]}}{\pa t_{m}}=(n+1)\left(L^{n}\frac{\pa L}{\pa t_{m}}\right)_{[0]}\\
 &=&(n+1)\left((B_{n})_{[0]}\frac{\pa u_1}{\pa t_{m}}+(B_{n})_{[1]}
 \frac{\pa u_2}{\pa t_{m}}\right)\\
 &=&(n+1)(F_{n,0}F_{m,1,0}+F_{m,0,0}F_{n,1})
   \eean
   where $u_1=F_{01}$ and $u_2=e^{F_{00}}$ have been used to reach the last equality.
   Similarly, we have
   \[
F_{n,m+1,0}=
(m+1)(F_{m,0}F_{n,1,0}+F_{n,0,0}F_{m,1}).
   \]
   Hence
   \[
\frac{F_{n+1,m}}{n+1}+\frac{F_{n,m+1}}{m+1}=F_{m,0}F_{n,1}+
F_{m,1}F_{n,0},\quad (n\geq 1,m\geq 1).
   \]
   Equations (\ref{F2}) and (\ref{F3}) can be verified in a similar manner.
\end{proof}
 \begin{corollary}
The two point functions $F_{mn}$, $F_{\wh{m}n}$, and
$F_{\wh{m}\wh{n}}$ are all determined by the fundamental variables
$F_{00}$ and $F_{01}$.
 \end{corollary}
 \begin{proof}
 This is just an immediate consequence of Proposition \ref{Fn01} and Theorem \ref{F123}.
 \end{proof}
 We shall show later on that the expression of (\ref{F1})-(\ref{F3})
has a simple interpretation from topological field theory.
\section{Catalan numbers and two-point functions}
 From dispersionless Hirota equations (\ref{F1})-(\ref{F3}), we see that
the building blocks are the two-point functions (\ref{bF1}) and
(\ref{bF2}). Motivated by the work of Kodama and Pierce \cite{KP09}
we like to consider the two-point functions $F_{mn}$, $F_{\wh{m}n}$, and
$F_{\wh{m}\wh{n}}$ in the case with $F_{00}=F_{01}=0$.
 \begin{proposition}
 \bea
 F_{2k,0}&=&(k+1)C_k,\quad F_{2k+1,0}=0
\label{B1-1}\\
F_{2k+1,1}&=&(2k+1)C_k,\quad F_{2k,1}=0 \label{B1-2}\\
F_{\wh{2k},0}&=&-(k+1)d_kC_k,\quad F_{\wh{2k+1},0}=0
\label{B2-1}\\
F_{\wh{2k+1},1}&=&-(2k+1)\left(d_k+\frac{1}{2(k+1)}\right)C_k,\quad
F_{\wh{2k},1}=0 \label{B2-2}
 \eea
where $C_k$ is the $k$-th Catalan number defined by
\[
C_k=\frac{1}{k+1}{2k\choose k}.
\]
 \end{proposition}
 \begin{proof}
This is an immediate consequence by setting $u_1=F_{01}=0$ and
$u_2=F_{11}=e^{{F_{00}}}=1$ in the equations
(\ref{Fn0})-(\ref{Fhn1}).
 \end{proof}
Let us derive the two point functions $F_{n,m}$ from the dHirota equation (\ref{F1}).
\begin{theorem} (Kodama and Pierce\cite{KP09})
The two point function $F_{nm}$ for the extended dToda hierarchy
with $F_{00}=0$ and $F_{01}=0$ are given by
 \bean
F_{2k,0}&=&(k+1)C_k,\quad k=1,2,\cdots\\
F_{2k+1,2l+1}&=&\frac{(2l+1)(2k+1)(l+1)(k+1)}{l+k+1}C_kC_l,\quad k,l=0,1,2,\cdots\\
F_{2k,2l}&=&\frac{lk(l+1)(k+1)}{l+k}C_kC_l\quad k,l=1,2,\cdots\\
F_{nm}&=&0,\quad {\rm otherwise\/}
  \eean
  where $C_k$ is the $k$-th Catalan number.
 \end{theorem}
\begin{proof}
Here we present a derivation of $F_{2k,2l}$ from the dHirota
equation (\ref{F1}). Writing $F_{2k,2l}$ in the expression
 \bean
  F_{2k,2l}&=&(F_{2k,2l}+\frac{2l}{2k+1}F_{2k+1,2l-1})
  -\frac{2l}{2k+1}(F_{2k+1,2l-1}+\frac{2l-1}{2k+2}F_{2k+2,2l-2})\\
&&+\frac{2l(2l-1)}{(2k+1)(2k+2)}(F_{2k+2,2l-2}+\frac{2l-2}{2k+3}F_{2k+3,2l-3})+\cdots\\
&&+\frac{2l(2l-1)\cdots 3}{(2k+1)(2k+2)\cdots (2k+2l-2)}(F_{2k+2l-2,2}+\frac{2}{2k+2l-1}F_{2k+2l-1,1})\\
&&-\frac{2l!}{(2k+1)(2k+2)\cdots (2k+2l-1)}F_{2k+2l-1,1}.
 \eean
Then, using  the dHirota equation (\ref{F1}), we have
 \bean
  F_{2k,2l}&=&2l(F_{2k,0}F_{2l-1,1}+F_{2k,1}F_{2l-1,0})
  -\frac{2l(2l-1)}{2k+1}(F_{2k+1,0}F_{2l-2,1}+F_{2k+1,1}F_{2l-2,0})\\
&&+\frac{2l(2l-1)(2l-2)}{(2k+1)(2k+2)}(F_{2k+2,0}F_{2l-3,1}+F_{2k+2,1}F_{2l-3,0})+\cdots\\
&&-\frac{2l!}{(2k+1)(2k+2)\cdots (2k+2l-1)}F_{2k+2l-1,1}.
 \eean
Taking into account  (\ref{B1-1}) and (\ref{B1-2}) we get
 \bean
  F_{2k,2l}&=&2l((k+1)C_k(2l-1)C_{l-1})-\frac{2l(2l-1)}{2k+1}((2k+1)C_klC_{l-1})\\
&&+\frac{2l(2l-1)(2l-2)}{(2k+1)(2k+2)}((k+2)C_{k+1}(2l-3)C_{l-2})+\cdots\\
&&+\frac{(2l)!}{(2k+1)(2k+2)\cdots (2k+2l-2)}((k+l)C_{k+l-1})\\
&&-\frac{(2l)!}{(2k+1)(2k+2)\cdots (2k+2l-1)}((2k+2l-1)C_{k+l-1})\\
&=&\sum_{i=0}^{l-1}\frac{(2l)!(2k)!}{(2l-2i-2)!(2k+2i)!}[(k+i+1)-(l-i)]C_{k+i}C_{l-i-1}\\
&=&\frac{l!k!(l+1)!(k+1)!C_kC_l}{(k+l)!(k+l-1)!}
 \sum_{i=0}^{l-1}{k+l-1\choose i}\left[{k+l\choose i+1}-{k+l\choose i}\right]
 \eean
where the formula
\[
C_{k+p}=2^p\frac{(2k+2p-1)!!(k+1)!}{(2k-1)!!(k+p+1)!}C_k
\]
for the Catalan numbers has been used. Since for any $p>1$, we have
 \bean
&&\sum_{i=0}^{p-1}{k+l-1\choose i}\left[{k+l\choose i+1}-{k+l\choose
i}\right]=\sum_{i=1}^{p}{k+l-1\choose i-1}\left[{k+l\choose
i}-{k+l\choose
i-1}\right]\\
&&=\sum_{i=2}^{p}{k+l-1\choose i-1}\left[{k+l-1\choose
i}+{k+l-1\choose i-1}-{k+l-1\choose i-1}-{k+l-1\choose
i-2}\right]+(k+l-1)\\
&&={k+l-1\choose p-1}{k+l-1\choose p}
 \eean
where the Pascal identity ${a\choose b}={a-1\choose b}+{a-1\choose b-1}$
has been used to reach the second equality.

Hence,
 \bean F_{2k,2l}
 &=&\frac{l!k!(l+1)!(k+1)!C_kC_l}{(k+l)!(k+l-1)!}{k+l-1\choose
l-1}{k+l-1\choose l}\\
&=&\frac{lk(l+1)(k+1)}{k+l}C_kC_l.
 \eean
 Substituting $F_{2k,2l}$ into (\ref{F1}) for $n=2k$ and $m=2l+1$, we obtain
 \[
F_{2k+1,2l+1}=\frac{(2l+1)(2k+1)(l+1)(k+1)}{k+l+1}C_kC_l.
 \]
 This is just the result obtained by Kodama and Pierce in \cite{KP09}.
\end{proof}
\begin{corollary}
The two point functions $F_{mn}$ for $m,n\geq 0$ are
positive-defined, i.e. $F_{mn}\geq 0$.
\end{corollary}
Next we deal with the two point function $F_{\wh{n}m}$ .
\begin{theorem}
The two point function $F_{\wh{n}m}$ for the extended dToda
hierarchy with $F_{00}=0$ and $F_{01}=0$ are given by
 \bean
F_{\wh{2k},0}&=&-(k+1)d_kC_k,\quad k=1,2,\cdots\\
F_{\wh{2k+1},2l+1}&=&-\frac{(2l+1)(2k+1)(l+1)(k+1)}{l+k+1}\left(d_k+\frac{1}{2(l+k+1)}\right)C_kC_l,
\quad k,l=0,1,2,\cdots\\
F_{\wh{2k},2l}&=&-\frac{lk(l+1)(k+1)}{l+k}\left(d_k-\frac{l}{2k(l+k)}\right)C_kC_j\quad k,l=1,2,\cdots\\
F_{\wh{n}m}&=&0,\quad {\rm otherwise\/}.
 \eean
\end{theorem}
\begin{proof}
Following the same procedure by using  the dHirota equation
(\ref{F2}) and (\ref{B1-1})-(\ref{B2-2}), we have
 \bean
F_{\wh{2k},2l}&=&-\sum_{i=0}^{l-1}\frac{(2l)!(2k)!C_{k+i}C_{l-i-1}}{(2l-2i-2)!(2k+2i)!}\left[(k+i+1)d_{k+i}-(l-i)
\left(d_{k+i}+\frac{1}{2(k+i+1)}\right)\right]\\
&=&(I)+(II)
 \eean
 where
 \bean
 (I)&=&-k!l!(k+1)!(l+1)!C_kC_l\sum_{i=0}^{l-1}\frac{(k-l+2i+1)d_k}{(l-i)!(l-i-1)!(k+i+1)!(k+i)!}\\
 (II)&=&-k!l!(k+1)!(l+1)!C_kC_l\sum_{i=1}^{l-1}\frac{(k-l+2i+1)(\frac{1}{k+i}+\cdots+\frac{1}{k+1})-
 \frac{(l-i)}{2(k+i+1)}}{(l-i)!(l-i-1)!(k+i+1)!(k+i)!}.
  \eean
Part (I)  can be computed as before and it gives
 \[
(I)=-\frac{lk(l+1)(k+1)}{k+l}d_kC_kC_l.
 \]
While part (II) can be written as follows
\[
-\frac{lk(l+1)(k+1)}{k+l}C_kC_l\left[\frac{2(k+l)\sum_{i=1}^{l-1}{k+l-1\choose i-1}[{k+l\choose i}-
{k+l\choose i-1}](\frac{1}{k+l-i}+\cdots+\frac{1}{k+1})-
\sum_{i=0}^{l-1}{k+l\choose i}^2}{2(k+l){k+l-1\choose l}{k+l-1\choose k}}\right]
\]
where the first summation of the numerator in the bracket can be
simplified as
 \bean
 &&\sum_{i=1}^{l-1}{k+l-1\choose i-1}\left[{k+l\choose i}-
{k+l\choose i-1}\right]\left(\frac{1}{k+l-i}+\cdots+\frac{1}{k+1}\right)\\
&&=\frac{1}{k+1}\sum_{i=1}^{l-1}{k+l-1\choose i-1}\left[{k+l\choose
i}-{k+l\choose i-1}\right]
+\frac{1}{k+2}\sum_{i=1}^{l-2}{k+l-1\choose i-1}\left[{k+l\choose i}-{k+l\choose i-1}\right]\\
&&+\cdots+\frac{1}{k+l-1}\sum_{i=1}^{1}{k+l-1\choose
i-1}\left[{k+l\choose i}-{k+l\choose i-1}\right]\\
&&=\frac{1}{k+1}{k+l-1\choose l-1}{k+l-1\choose
l-2}+\frac{1}{k+2}{k+l-1\choose l-2}{k+l-1\choose
l-3}\\
&&+\cdots+\frac{1}{k+l-1}{k+l-1\choose 1}{k+l-1\choose 0}\\
&&=\sum_{i=1}^{l-1}\frac{1}{k+l-i}{k+l-1\choose i-1}{k+l-1\choose
i}.
 \eean
Hence the numerator in the bracket is given by
 \bean
 &&2(k+l)\sum_{i=1}^{l-1}\frac{1}{k+l-i}{k+l-1\choose i-1}{k+l-1\choose i}-\sum_{i=0}^{l-1}{k+l\choose i}^2\\
 &&=\sum_{i=1}^{l-1}\left[2{k+l-1\choose i-1}{k+l\choose i}-{k+l\choose i}^2\right]-{k+l\choose 0}^2\\
  &&=\sum_{i=1}^{l-1}{k+l\choose i}\left[2{k+l-1\choose i-1}-{k+l-1\choose i}-{k+l-1\choose i-1}\right]+1\\
   &&=\sum_{i=1}^{l-1}\left[{k+l-1\choose i}+{k+l-1\choose i-1}\right]
   \left[{k+l-1\choose i-1}-{k+l-1\choose i}\right]+1\\
   &&=-{k+l-1\choose l-1}^2
 \eean
 which implies
 \[
F_{\wh{2k},2l}=-\frac{lk(l+1)(k+1)}{l+k}\left(d_k-\frac{l}{2k(l+k)}\right)C_kC_j.
 \]
 Substituting $F_{\wh{2k},2l}$ into (\ref{F2}) for $n=2k$ and $m=2l+1$, we obtain
 \[
F_{\wh{2k+1},2l+1}=-\frac{(2l+1)(2k+1)(l+1)(k+1)}{k+l+1}\left(d_k+\frac{1}{2(l+k+1)}\right)C_kC_l.
 \]
\end{proof}
\begin{corollary}
The two point functions $F_{\wh{m}n}$ for $mn\neq 0$ are
negative-defined, i.e. $F_{\wh{m}n}< 0$.
\end{corollary}
\begin{proof}
The only case to be considered is $F_{\wh{2k},2l}$ in which
\[
d_k-\frac{l}{2k(l+k)}=\left(d_k-\frac{1}{2k}\right)+\frac{1}{2(k+l)}>0.
\]
\end{proof}
Finally, we compute the two point function $F_{\wh{n}\wh{m}}$.
\begin{theorem}
The two point function $F_{\wh{n}\wh{m}}$ for the extended dToda
hierarchy with $F_{00}=0$ and $F_{01}=0$ are given by
 \bean
F_{\wh{2k},\wh{0}}&=&-\frac{k+1}{2}d_kC_k,\quad k=1,2,\cdots\\
F_{\wh{2k+1},\wh{2l+1}}&=&\frac{(2l+1)(2k+1)(l+1)(k+1)}{l+k+1}\times\\
&&\left[\left(d_k+\frac{1}{2(l+k+1)}\right)\left(d_l+\frac{1}{2(l+k+1)}\right)+\frac{1}{4(k+l+1)^2}\right]C_kC_l,\quad k,l=0,1,2,\cdots\\
F_{\wh{2k},\wh{2l}}&=&\frac{lk(l+1)(k+1)}{l+k}\left[\left(d_k-\frac{l}{2k(l+k)}\right)\left(d_l-\frac{k}{2l(l+k)}\right)+
\frac{1}{4(k+l)^2}\right]C_kC_l,\quad k,l=1,2,\cdots\\
F_{\wh{n}\wh{m}}&=&0,\quad {\rm otherwise\/}.
 \eean
\end{theorem}
\begin{proof}
 Using  the dHirota equation (\ref{F3}) and taking into
account (\ref{B2-1})-(\ref{B2-2}) we have
 \bean
  F_{\wh{2k},\wh{2l}}&=&\sum_{i=0}^{l-1}\frac{(2l)!(2k)!C_{k+i}C_{l-i-1}}{(2l-2i-2)!(2k+2i)!}\times\\
&&\left[(k+i+1)d_{k+i}(d_{l-i-1}+\frac{1}{2(l-i)})-(d_{k+i}+\frac{1}{2(k+i+1)})(l-i)d_{l-i-1}\right]\\
&&-\frac{2l!2k!}{2(2k+2l)!}(k+l+1)d_{k+l}C_{k+l}\\
&=&\frac{l!k!(l+1)!(k+1)!C_kC_l}{2((k+l)!)^2}[(I)+(II)+(III)+(IV)+(V)]
 \eean
 where
 \[
(I)=2((k+l)!)^2\sum_{i=0}^{l-1}
 \frac{\left[(k-l+2i+1)(d_k+(\frac{1}{k+i}+\cdots+\frac{1}{k+1}))-\frac{(l-i)}{2(k+i+1)}\right]d_l}
{(l-i)!(l-i-1)!(k+i+1)!(k+i)!}
\]
\[
 (II)=2((k+l)!)^2\sum_{i=0}^{l-1}
 \frac{\left[-(k-l+2i+1)(\frac{1}{l-i}+\cdots+\frac{1}{l})+\frac{(k+i+1)}{2(l-i)}\right]d_k}
{(l-i)!(l-i-1)!(k+i+1)!(k+i)!}\\
 \]
 \[
 (III)=2((k+l)!)^2\sum_{i=0}^{l-1}
 \frac{\left[-(k-l+2i+1)(\frac{1}{k+i}+\cdots+\frac{1}{k+1})(\frac{1}{l-i}+\cdots+\frac{1}{l})\right]}
{(l-i)!(l-i-1)!(k+i+1)!(k+i)!}
 \]
 \[
(IV)=2((k+l)!)^2\sum_{i=0}^{l-1}
 \frac{\left[\frac{(k+i+1)}{2(l-i)}(\frac{1}{k+i}+\cdots+\frac{1}{k+1})+
 \frac{(l-i)}{2(k+i+1)}(\frac{1}{l-i}+\cdots+\frac{1}{l})\right]}
{(l-i)!(l-i-1)!(k+i+1)!(k+i)!}
 \]
 \[
(V)=-d_{k+l}.
 \]
    Each term can be calculated as follows:
 \bean
(I)&=&2(k+l){k+l-1\choose k}{k+l-1\choose l}\left[d_k-\frac{l}{2k(l+k)}\right]d_l\\
 (II)&=&d_k-2(k+l){k+l-1\choose k}{k+l-1\choose l}\left[\frac{d_kk}{2l(k+l)}\right]\\
 (III)&=&-\sum_{j=0}^{l-1}\frac{1}{j+1}\sum_{i=0}^j{k+l\choose i}^2
-\sum_{j=1}^{l-1}\frac{1}{k+l-j}\sum_{i=1}^{j}{k+l\choose i}^2\nonumber\\
&&+\frac{1}{k+l}{k+l-1\choose l-1}{k+l-1\choose l}+\sum_{i=0}^{l-1}\frac{1}{k+l-i}\\
(IV)&=&\sum_{j=1}^{l-1}\frac{1}{k+l-j}\sum_{i=1}^{j}{k+l\choose i}^2
+\sum_{j=0}^{l-1}\frac{1}{j+1}\sum_{i=0}^{j}{k+l\choose i}^2\\
 (V)&=&-d_k-\sum_{i=0}^{l-1}\frac{1}{k+l-i}.
  \eean
  It turns out that
  \[
  F_{\wh{2k},\wh{2l}}=\frac{lk(l+1)(k+1)}{l+k}\left[\left(d_k-\frac{l}{2k(l+k)}\right)
  \left(d_l-\frac{k}{2l(l+k)}\right)+\frac{1}{4(k+l)^2}\right]C_kC_l.
  \]
  Substituting $F_{\wh{2k},\wh{2l}}$ into (\ref{F3}) for $n=2k$ and $m=2l+1$, we obtain
 \bean
F_{\wh{2k+1},\wh{2l+1}}&=&\frac{(2l+1)(2k+1)(l+1)(k+1)}{l+k+1}\times\\
&&\left[\left(d_k+\frac{1}{2(l+k+1)}\right)\left(d_l+\frac{1}{2(l+k+1)}\right)+\frac{1}{4(k+l+1)^2}\right]C_kC_l.
 \eean
\end{proof}
\begin{corollary}
The two point functions $F_{\wh{m}\wh{n}}$ for $mn\neq
0$ are positive-defined, i.e. $F_{\wh{m}\wh{n}}\geq 0$.
\end{corollary}
\section{Back to the topological $CP^1$ model}
The relationship between integrable systems and topological field theories has dramatic advances
in the past two decades(see e.g. \cite{AK96,Dij93,Dub96,Kon92,Kri92,W90,W91}).
For the extended dToda hierarchy the corresponding topological field is described by
two primary fields (or observables) $\{\mathcal{O}_1=1\in H^0(CP^1), \mathcal{O}_2=\om\in H^2(CP^1)\}$
with coupling parameters $T^{\alpha,0}, \alpha=1,2$.
When the theory couples to topological gravity, a set of new variables emerge as gravitational descendants
$\{\sigma_n(\mathcal{O}_\alpha)\}$ with new coupling constants $\{T^{\alpha,n}\}$. The identity operator
 now becomes the puncture operator $\mathcal{O}_1=P$ and we also denote $\mathcal{O}_2=Q$.
The space spanned by $\{T^{\alpha,n}, n=0,1,2,\cdots\}$ is called the full phase space and the subspace
parametrized by $T^{\alpha,0}$ the small phase space.
The generating function of correlation function is the full free energy defined by
\[
\mathcal{F}(T)=\sum_{g=0}\mathcal{F}_g=
\sum_{g=0}^\infty\langle e^{\sum_{\alpha,n}T^{\alpha,n}\sigma_n(\mathcal{O}_\alpha)}\rangle_g.
\]
Since the free energy $F(t_0,t,\hat{t})$ of the extended dToda hierarchy corresponds to the genu-zero
generating function $\mathcal{F}_0$ of $CP^1$ under the identification
\[
t_{n+1}=\frac{T^{2,n}}{(n+1)!},\quad
\hat{t}_n=\frac{2T^{1,n}}{n!},\quad n\ge 0.
 \]
 where $\hat{t}_0=2t_0=2T^{1,0}=2x$. Hence a generic genus-zero $m$-point correlation
 function can be calculated as follows
\[
\langle\sigma_{n_1}(\mathcal{O}_{\alpha_1})\sigma_{n_2}(\mathcal{O}_{\alpha_2})
\cdots\sigma_{n_m}(\mathcal{O}_{\alpha_m})\rangle
=\frac{\pa^mF}{\pa T^{\alpha_1,n_1}\pa T^{\alpha_2,n_2}\cdots\pa T^{\alpha_m,n_m} }.
\]
In particular, the metric  on the space of primary fields is defined by
three-point correlation function $\eta_{\alpha\beta}=\langle P\mathcal{O}_\alpha\mathcal{O}_\beta\rangle$ with
$\eta_{11}=\eta_{22}=0$ and $\eta_{12}=\eta_{21}=1$, and hence
$\mathcal{O}_1=\mathcal{O}^2$ and $\mathcal{O}_2=\mathcal{O}^1$.

The Lax equations of the extended dToda hierarchy can be written as
 \[
\frac{\pa L}{\pa T^{\alpha,n}}=\{B_{\alpha,n},L\},\quad \alpha=1,2; n=0,1,2,\cdots
 \]
 where
 \[
B_{1,n}=\frac{2}{n!}(L^n(\log L-d_n))_{\ge 0},\quad
B_{2,n}=\frac{1}{(n+1)!}(L^{n+1})_{\ge 0}
 \]
and the string equation (\ref{string-eq}) becomes
 \[
  0=1+\sum_{n=1}^{\infty}T^{2,n}\frac{\pa L}{\pa
T^{2,n-1}} + \sum_{n=1}^{\infty}T^{1,n}\frac{\pa L}{\pa T^{1,n-1}}.
 \]
 Shifting the variable $T^{1,1}\to T^{1,1}-1$ we have
  \[
\frac{\pa L}{\pa T^{1,0}}=1+\sum_{n=1}^{\infty}T^{2,n}\frac{\pa
L}{\pa T^{2,n-1}} + \sum_{n=1}^{\infty}T^{1,n}\frac{\pa L}{\pa
T^{1,n-1}}
 \]
 which, after extracting the $p^0$ term, yields
 \bean
t^1(T)&=&T^{1,0}+\sum_\alpha\sum_{n=1}^{\infty}T^{\alpha,n}
\langle\sigma_{n-1}(\mathcal{O}_\alpha)Q\rangle,\\
t^2(T)&=&T^{2,0}+\sum_\alpha\sum_{n=1}^{\infty}T^{\alpha,n}
\langle\sigma_{n-1}(\mathcal{O}_\alpha)P\rangle
 \eean
  where we identify the flat coordinate $t^\alpha=\langle P\mathcal{O}^\alpha\rangle$ as
 \[
  t^1=u_1=\langle PQ\rangle,\quad t^2=\log u_2=\langle PP\rangle.
   \]
Therefore, in small space $t^1=T^{1,0}$ and $ t^2=T^{2,0}$. The condition $F_{01}=F_{00}=0$ then corresponds to
 $T^{\alpha,n}=0,\; \forall \alpha, n$.
In the Landau-Ginzburg formulation of the topological $CP^1$ model, it can be shown \cite{W90} that
 the following genus-zero topological recursion relation holds.
 \be
 \langle
\sigma_n(\mathcal{O}_\alpha)XY\rangle =\sum_\beta
\langle\sigma_{n-1}(\mathcal{O}_\alpha)\mathcal{O}^\beta\rangle
\langle\mathcal{O}_\beta XY\rangle.
 \label{TRR}
 \ee
 \begin{proposition}
The genus-zero topological recursion relation (\ref{TRR}) implies the dHirota equations (\ref{F1})-(\ref{F3}).
 \end{proposition}
\begin{proof}
Using (\ref{TRR}) we have
\bean
&&\frac{\pa}{\pa T^{1,0}}[\langle\sig_{n+1}(Q)\sig_m(Q)\rangle+\langle\sig_{n}(Q)\sig_{m+1}(Q)\rangle]\\
&&=\langle\sig_{n+1}(Q)\sig_m(Q)P\rangle+\langle\sig_{n}(Q)\sig_{m+1}(Q)P\rangle\\
&&=\frac{\pa}{\pa T^{1,0}} [\langle\sig_{n}(Q)Q\rangle\langle \sig_m(Q)P\rangle+\langle\sig_n(Q)P\rangle\langle\sig_m(Q)Q\rangle]
\eean
which, after integrating over $T^{1,0}$, implies
\[
\langle\sig_{n+1}(Q)\sig_m(Q)\rangle+\langle\sig_{n}(Q)\sig_{m+1}(Q)\rangle=
 \langle\sig_{n}(Q)Q\rangle\langle \sig_m(Q)P\rangle+\langle\sig_n(Q)P\rangle\langle\sig_m(Q)Q\rangle.
\]
Similarly, we have
\bean
&&\langle\sig_{n+1}(P)\sig_m(Q)\rangle+\langle\sig_{n}(P)\sig_{m+1}(Q)\rangle=
 \langle\sig_{n}(P)Q\rangle\langle \sig_m(Q)P\rangle+\langle\sig_n(P)P\rangle\langle\sig_m(Q)Q\rangle,\\
&&\langle\sig_{n+1}(P)\sig_m(P)\rangle+\langle\sig_{n}(P)\sig_{m+1}(P)\rangle=
\langle\sig_{n}(P)Q\rangle\langle\sig_m(P)P\rangle+\langle\sig_n(P)P\rangle\langle\sig_m(P)Q\rangle.
 \eean
The proof is completed by noting the following identifications:
 \bean
 \langle\sig_{m}(P)\sig_{n}(P)\rangle&=&\frac{4F_{\hat{m}\hat{n}}}{m!n!},\quad m,n\geq 0\\
\langle\sig_{m}(P)\sig_{n-1}(Q)\rangle&=&\frac{2F_{\hat{m}n}}{m!n!},\quad m\geq 0, n\geq 1\\
\langle\sig_{m-1}(Q)\sig_{n-1}(Q)\rangle&=&\frac{F_{mn}}{m!n!},\quad m,n\geq 1.
 \eean
\end{proof}
We thus show that the integrable structure associated with the genus-zero topological $CP^1$ model
is the extended dToda hierarchy.
Furthermore, integrating the two point functions
$\langle\sigma_n(P)P\rangle$ and $\langle\sigma_n(Q)P\rangle$ over
$T^{1,0}$ we obtain the one-point functions
\[
\langle\sigma_n(P)\rangle=\frac{2}{(n+1)!}F_{\wh{n+1},0},\quad
\langle\sigma_n(Q)\rangle=\frac{1}{(n+2)!}F_{n+2,0}.
\]
In particular, their values in the limit of zero couplings
($T^{\alpha,n}=0$ $\forall \alpha,n$) are
\[
\langle\sigma_{2k-1}(P)\rangle=-\frac{2d_k}{(k!)^2},\quad
\langle\sigma_{2k-2}(Q)\rangle=\frac{1}{(k!)^2}.
\]
\section{Concluding remarks}

We have introduced the extended dToda hierarchy from the one-dimensional dToda hierarchy
by  adding logarithmic flows. The full hierarchy equations of the
extended dToda system can be summarized by a set of dHirota equations which involve second derivatives of the free energy
$F$ in time parameters $t_0$, $t_n$ and $\hat{t}_n$. Based on these dHirota equations we computed the
two point functions $F_{n,m}$, $F_{\hat{n},m}$, and $F_{\hat{n},\hat{m}}$ in the case with $F_{00}=F_{01}=0$.
Our results extend the previous formula obtained by Kodama and Pierce for the one-dimensional dToda system to
those results for the extended dToda system. Furthermore, we have shown that, in terms of $CP^1$ time parameters, the
dHirota equations are nothing but a direct consequence of the genus-zero topological recursion relations.
This provides another route to realize that the integrable structure associated with the topological $CP^1$ model at
genus-zero level is the extended dToda hierarchy.

There are two remarks in order.
First, Milanov \cite{Mil07} has studied the Hirota quadratic equations associated with the
extended Toda hierarchy  by constructing some vertex operators taking values
in the algebra of differential operators on the affine line. The peculiar properties of these
Hirota equations have been studied in some recent works \cite{LHWC10,T10}.
It would be interesting to investigate the dispersionless limit of the Hirota quadratic equations.
Second, in \cite{KP09} a combinatorial meaning of the two point functions $F_{nm}$ has been investigated
from large-$N$ expansion of unitary ensemble of random matrices. It is quite natural to ask how
to realize the geometric/topological meaning of the rational numbers $F_{\hat{n},m}$ and $F_{\hat{n},\hat{m}}$
from the $CP^1$ matrix integral\cite{EY94} which contains extra logarithmic terms.
We hope to back to all these issues in our future works.

{\bf Acknowledgments\/}\\
We like to thank H.F.Shen for useful discussions.
This work is partially supported by the National Science Council of
Taiwan under Grant No. NSC99-2115-M-167-001(NCL) and NSC100-2112-M-194-002-MY3(MHT).


\end{document}